\documentclass[12pt,psamsfonts]{amsart}
\usepackage{amsmath,amsthm,amsfonts,amssymb}
\usepackage{eucal}
\usepackage{graphicx}
\usepackage{xcolor}
\usepackage[all,knot]{xy}
\usepackage[pdftex,plainpages=false,hypertexnames=false,pdfpagelabels]{hyperref}
\definecolor{dark-red}{rgb}{0.7,0.25,0.25}
\definecolor{dark-blue}{rgb}{0.15,0.15,0.55}
\definecolor{medium-blue}{rgb}{0,0,.8}
\definecolor{DarkGreen}{RGB}{0,150,0}
\definecolor{rho}{named}{red}
\hypersetup{
   colorlinks, linkcolor={purple},
   citecolor={medium-blue}, urlcolor={medium-blue}
}
\xyoption{arc}

\usepackage{fullpage}

\newcommand{\Z}{\mathbb{Z}}

\newcommand{\C}{\mathbb{C}}
\newcommand{\bbX}{\mathbb{X}}
\newcommand{\bbH}{\mathbb{H}}

\newcommand{\frj}{\mathfrak{j}}

\newcommand{\GL}{\operatorname{GL}}
\newcommand{\U}{\operatorname{U}}
\newcommand{\SL}{\operatorname{SL}}
\newcommand{\SU}{\operatorname{SU}}

\newcommand{\PSL}{\operatorname{PSL}}
\newcommand{\g}{\mathfrak{g}}
\newcommand{\ch}{\operatorname{ch}}
\newcommand{\Span}{\operatorname{span}}
\newcommand{\Rep}{\operatorname{Rep}}
\newcommand{\Tr}{\operatorname{Tr}}
\renewcommand{\Re}{\operatorname{Re}}
\newcommand{\rank}{\operatorname{rank}}

\newcommand{\Ising}{\operatorname{Ising}}
\newcommand{\oVec}{\operatorname{Vec}}

\newcommand{\mcA}{\mathcal{A}}
\newcommand{\mcB}{\mathcal{B}}
\newcommand{\mcC}{\mathcal{C}}
\newcommand{\mcD}{\mathcal{D}}

\newcommand{\mcG}{\mathcal{G}}
\newcommand{\mcH}{\mathcal{H}}

\newcommand{\mcM}{\mathcal{M}}

\newcommand{\mcP}{\mathcal{P}}

\newcommand{\mcV}{\mathcal{V}}
\newcommand{\mcW}{\mathcal{W}}

\newcommand{\abs}[1]{\left| #1 \right|}

\numberwithin{equation}{section}

\newtheorem*{thm}{Theorem}

\newtheorem{theorem}{Theorem}[section]

\newtheorem{conjecture}[theorem]{Conjecture}

\theoremstyle{definition}

\newtheorem{remark}[theorem]{Remark}

\newtheorem{question}[theorem]{Question}

\newtheorem{definition}[theorem]{Definition}

\begin{document}
\title[]{On classification of extremal non-holomorphic conformal field theories}

\author{James E. Tener}
\email{jtener@math.ucsb.edu}
\address{Dept. of Mathematics\\
    University of California\\
    Santa Barbara, CA 93106-6105\\
    U.S.A.}

\author{Zhenghan Wang}
\email{zhenghwa@microsoft.com}
\address{Microsoft Station Q and Dept. of Mathematics\\
    University of California\\
    Santa Barbara, CA 93106-6105\\
    U.S.A.}

\thanks{Both authors thank Terry Gannon for sharing his insight on vector-valued modular forms and conformal field theories.  The second author thanks Meng Cheng for helpful discussions, and is partially supported by NSF grants DMS-1410144 and DMS-1411212.}

\begin{abstract}

Rational chiral conformal field theories are organized according to their genus, which consists of a modular tensor category $\mcC$ and a central charge $c$.
A long-term goal is to classify unitary rational conformal field theories based on a classification of unitary modular tensor categories.
We conjecture that for any unitary modular tensor category $\mcC$, there exists a unitary chiral conformal field theory $\mcV$ so that its modular tensor category $\mcC_\mcV$ is $\mcC$. 
In this paper, we initiate a mathematical program in and around this conjecture.
We define a class of \emph{extremal} vertex operator algebras with minimal conformal dimensions as large as possible for their central charge, and non-trivial representation theory.
We show that there are finitely many different characters of extremal vertex operator algebras $\mcV$ possessing at most three different irreducible modules.
Moreover, we list all of the possible characters for such vertex operator algebras with $c \le 48$.

\end{abstract}
\maketitle

\section{Introduction}

Modeling and classification of topological phases of matter $\mcH$ is an interesting and difficult mathematical problem.
The theory of topological phases of matter is most mature in two spatial dimensions, where topological excitations of a topological phase $\mcH$ are modeled by a unitary modular tensor category (UMC) $\mcC_\mcH$ \cite{FKLW, WangBook}.
However, how to model the boundary physics of a topological phase is still subtle.  
It is widely believed in the case of fractional quantum Hall states that the boundary physics, which should be described by some quantum field theory, can be modeled by a unitary chiral conformal field theory ($\chi$CFT) $\mcV$ \cite{Wen, Read}.  
As an instance of a bulk-edge correspondence, the UMC $\mcC_\mcV$ encoded in the boundary $\chi$CFT $\mcV$ is the same as the UMC $\mcC_\mcH$ of the bulk.
Moreover, the UMC $\mcC_\mcH$ has a multiplicative central charge $\chi=e^{\pi ic/4}$, where $c$, called the (chiral) topological central charge of $\mcC_\mcH$, is a non-negative rational number defined modulo $8$, which 
agrees with the central charge of $\mcV$.

We conjecture that this bulk-edge correspondence exists beyond the fractional quantum Hall states, so that for any given UMC $\mcC$, there is always a unitary $\chi$CFT $\mathcal{V}$ such that its modular tensor category $\mcC_\mcV$ is $\mcC$ and its central charge is equal to the topological central charge of $\mcC$, modulo $8$.  The same conjecture was made by Gannon \cite{GannonOaxaca} as an analogue of Tannaka-Krein duality.  
In this paper, we initiate a mathematical program in and around this conjecture based on the theory of vector-valued modular forms (see \cite{Gannon14} and references therein). 
A long-term goal is to classify unitary $\chi$CFTs based on progress in classifying UMCs \cite{BNRW1,RSW}.
However, the bulk-edge correspondence between UMCs and $\chi$CFTs is a subtle one.  
First, the correspondence does not behave well with respect to tensor products: the UMC associated to $\SU(2)_5$ splits as the tensor product of two UMCs, but the corresponding $\chi$CFT does not split\footnote{The second author lost a bet to Nathan Seiberg on this example in 2015.}.  
Secondly, it is tempting to conjecture that there is a unique CFT in the realizable genus of minimal central charge for a fixed UMC, which turns out to also be wrong \cite{Cano}. 

Mathematically rigorously defined quantum field theories are relatively rare, with two large such classes being topological quantum field theories (TQFTs) and conformal field theories (CFTs).  
In $(2\!+\!1)$-dimensions, a TQFT is essentially the same as a modular tensor category \cite{TuraevBook}.
The chiral algebra of a CFT also gives rise to a modular tensor category, and we are interested in how much more data is needed to reconstruct a CFT from this modular tensor category.  
As explained above, this approach to CFTs is natural for the modeling of topological phases of matter.  

Two other considerations \cite{Schellekens93, Hoehn03} led to similar approaches.
In \cite{Schellekens93}, one starts with the observation that any $\chi$CFT has two naturally associated algebras--the chiral algebra and the fusion algebra.
Minimal chiral algebras lead to the class of minimal model CFTs, and minimal fusion algebras to holomorphic CFTs, which include the famous Moonshine Module.
In \cite{Hoehn03}, H\"ohn generalizes the notion of genus from lattices to vertex operator algebras (VOAs), mathematical objects which correspond to chiral algebras.
The genus of a nice VOA $\mcV$ is the pair $(\Rep(\mcV),c)$, where $\Rep(\mcV)$ is the modular tensor category of its representations and $c$ its central charge.

We will also only consider chiral CFTs in our approach to classification, of which there are two well-developed mathematical definitions: VOAs and local conformal nets.
While these two notions are conjecturally equivalent (see \cite{Yasu, Ten16b}), we will use VOAs as our mathematical $\chi$CFT.
It is still out of reach to classify vertex operator algebras in a given genus, and so we will instead focus on the classification of their characters.  
This, too, is a difficult question, and so our problem becomes to specify what extra data we should consider, along with the genus, to help classify characters.

The minimal energies (or minimal conformal weights) $\{h_i\}$ of a nice VOA are encoded, mod 1, in the exponents of the topological twists of its modular tensor category by $\theta_i = e^{2 \pi i h_i}$. 
Therefore, one set of natural extra data to consider would be a lifting of the exponents of the topological twists from equivalence classes of rational numbers (modulo $1$) to actual rational numbers.
Since we are mainly interested in unitary theories, we only consider liftings for which $h_i \ge 0$ for all $i$.
It is not impossible that a consistent lifting of the topological twists is sufficient to determine a corresponding CFT within a given genus, at least when the CFT has non-trivial representation theory (i.e. when the CFT is not holomorphic).

In terms of this data, we define an \emph{extremal} non-holomorphic VOA (Definition \ref{defExtremal}) to be one for which the sum $\sum h_i$ is as large as possible for its genus\footnote{H\"ohn introduced in \cite{Hoehn95} a notion of extremal \emph{holomorphic} VOA which is similar in spirit and function, but not directly compatible with our notion defined for non-holomorphic VOAs.}.
Restricting to the class of extremal VOAs allows one to prove results about classification by working modulo the extremely difficult problem of classification of holomorphic CFTs.
Our main result, which appears in the body of the paper as Theorem \ref{thmExtremalConjectureVerified}, is an example of such a classification result:

\begin{thm}
Let $\mcC$ be a unitary modular tensor category with two or three simple objects, and let $c$ be a lifting of its topological central charge to a positive rational number.
Then the character vector of an extremal VOA $\mcV$ with genus $(\mcC, c)$ is uniquely determined by the minimal energies $h_i$ of its modules.
In particular, there are finitely many different character vectors of extremal VOAs in these genera.

If $\rank(\mcC) = 2$ and $c \le 72$, then the characters of all extremal VOAs in the genus $(\mcC, c)$ are given in the tables of Section \ref{secRankTwo}.
If $\rank(\mcC) = 3$ and $c \le 48$, then the characters of all extremal VOAs in the genus $(\mcC, c)$ are given in the tables of Section \ref{secRankThree}.
\end{thm}
The primary technical tool for proving Theorem \ref{thmExtremalConjectureVerified} is the beautiful theory of vector-valued modular forms, in particular the approach developed in \cite{Mason07,Gannon14}, and the method of computing fundamental matrices described in \cite{BG}.
Our work continues the systematic study of VOA characters in small rank which began with \cite{Junla}.

The tables of characters of extremal VOAs presented in Sections \ref{secRankTwo} and Sections \ref{secRankThree} include the characters of every WZW model with the appropriate number of irreducible modules and appropriate central charge.
They also include the characters of the Ising minimal model, the baby monster VOA $V\!B_{(0)}^\natural$, and several examples of characters of VOAs recently constructed as cosets in \cite{GHM}.
Finally, there are several intriguing examples of vector valued modular forms with positive integer coefficients for which no VOA character realization is known, and which cannot be a non-trivial linear combination of characters of VOAs.

\section{Vertex operator algebras, characters, and vector-valued modular forms}

\subsection{VOAs and their genera}

Let $\mcV$ be a vertex operator algebra (VOA) with central charge $c$.
Here, and throughout the remainder of the paper, we will assume that $\mcV$ is rational, $C_2$-cofinite, and of CFT type.
Then $\mcV$ has finitely many isomorphism classes of modules, for which we choose representatives $\mcV = M_0, M_1, \ldots, M_{d-1}$.

The characters of $M_i$ are given by
$$
\ch M_i(\tau) = q^{-c/24+h_i} \sum_{n \ge 0} \dim M_i(n+h_i) \,\,q^n,
$$
where $h_i$ is the smallest eigenvalue of the energy operator $L_0$ on $M_i$, and $M_i(n+h_i)$ is the eigenspace of $L_0$ corresponding to the eigenvalue $n + h_i$. 
Since we are primarily interested in unitary theories, we will assume throughout that $c > 0$ and that $h_i > 0 $ for $i > 0$.
As always, we have set $q = e^{2 \pi i \tau}$.

It was proven by Zhu \cite{Zhu96} that the character vector 
$$
\bbX = \begin{pmatrix} \ch M_0 \\ \vdots \\ \ch M_{d-1} \end{pmatrix}
$$
is a vector-valued modular function for a certain representation $\rho:\SL(2,\Z) \to \GL(d)$ constructed from $\mcV$.
That is, $\bbX:\bbH \to \C^d$ is a holomorphic function on the upper half-plane which satisfies
\begin{equation}\label{eqModularFunction}
\bbX(\gamma \cdot \tau) = \rho(\gamma) \bbX(\tau)
\end{equation}
for all $\gamma \in \SL(2,\Z)$.

By the work of Huang and Lepowski \cite{Huang05}, the category $\Rep(\mcV)$ of $\mcV$-modules is a modular tensor category.
The modular data of $\Rep(\mcV)$ induces a representation of $\SL(2,\Z)$, which coincides with the representation $\rho$ used by Zhu\footnote{More precisely, to get an honest representation of $\SL(2,\Z)$ one sets
$$
\rho \begin{pmatrix} 0 & -1\\ 1 & 0 \end{pmatrix} = \frac{1}{\sqrt{\dim(\Rep(\mcV))}} \big(\Tr(\beta_{j,i}\beta_{i,j})\,\big)_{i,j}, \quad
\rho \begin{pmatrix} 1 & 1\\ 0 & 1 \end{pmatrix} = e^{-2 \pi i c/24} \big( \delta_{i,j} \theta_i \big)_{i,j}
$$
}.
The minimal energies $h_i$ can be recovered, mod $1$, from the data of the modular tensor category $\Rep(\mcV)$ from the fact that the twists of simple objects are given by
$$
\theta_i = e^{2 \pi i h_i}.
$$
The central charge $c$ can also be recovered, mod $8$, from $\Rep(\mcV)$ from the fact that
$$
e^{i \pi c/4} = \chi_{\Rep(\mcV)},
$$
where $\chi_{\mcC}$ is the multiplicative central charge of a modular tensor category $\mcC$.
This motivates the following definition.

\begin{definition}\label{defGenus}
Let $\mcC$ be a modular tensor category with multiplicative central charge $\chi_\mcC$, and let $\mcV$ be a (rational, $C_2$-cofinite, CFT type) VOA with central charge $c$.
\begin{enumerate}
\item \cite{Hoehn03} The \emph{genus} $\mcG(\mcV)$ of $\mcV$ is the pair $(\Rep(\mcV), c)$.
\item An \emph{admissible genus} is a pair $(\mcC,c^\prime)$ such that $\chi_\mcC=e^{i\pi c^\prime/4}$.
\item If $\mcC \cong \Rep(\mcV)$, then we say that $\mcV$ is a \emph{realization} of the admissible genus $(\mcC, c)$. 
We say that $(\mcC, c)$ is \emph{realizable} if it admits a realization.
\end{enumerate}
\end{definition}
We will call a VOA \emph{holomorphic} if $\Rep(\mcV)$ is equivalent to the trivial modular tensor category $\oVec$, in which case $c \in \{8, 16, 24, \ldots \}$.

\begin{conjecture}\label{cnjGenusFiniteness}
Let $\mcC$ be a unitary modular tensor category.
\begin{enumerate}
\item (Existence)  There is a central charge $c$ such that the admissible genus $(\mcC, c)$ is realizable.
\item (Finiteness) \cite{Hoehn03} Each realizable genus $(\mcC, c)$ is realized by a finite number of isomorphism classes of VOAs. 
\end{enumerate}
\end{conjecture}
There are few examples of genera for which the finiteness component of Conjecture \ref{cnjGenusFiniteness} has been verified; in particular, it remains open for $(\oVec, 24)$.
We will instead consider a weaker version of genus finiteness:

\begin{conjecture}\label{cnjCharacterFiniteness}
For each admissible genus $(\mcC, c)$ there are finitely many character vectors of VOA realizations.
\end{conjecture}

When $\mcC=\oVec$, Conjecture \ref{cnjCharacterFiniteness} has been verified for $c \in \{8,16,24\}$.
For both $c = 8$ or $16$ there is a unique character vector, realized by lattice model(s).
For $c=24$, there are 71 character vectors given by Schellekens' list \cite{Schellekens93}, which have all been shown to correspond to at least one VOA after considerable time and effort.
The genus finiteness problem is open for $c=32$.

\subsection{Extremal VOAs}

Given the difficulty of classifying holomorphic VOAs, we will take up an orthogonal problem.
If $\mcV$ and $\mcW$ are VOAs with $\mcW$ holomorphic, then $\Rep(\mcV \otimes \mcW) \cong \Rep(\mcV)$.
Thus there is an action of holomorphic VOAs on the class of VOAs with $\Rep(\mcV)$ equivalent to a fixed modular tensor category $\mcC$.
One can therefore try to understand this class modulo the action of holomorphic VOAs.

If the central charges of $\mcV$ and $\mcW$ are $c_\mcV$ and $c_\mcW$, respectively, then the central charge of $\mcV \otimes \mcW$ is $c_\mcV + c_\mcW$.
On the other hand, if $M_0, \ldots M_{d-1}$ are the irreducible modules of $\mcV$, then the irreducible modules of $\mcV \otimes \mcW$ are $M_0 \otimes \mcW, \ldots, M_{d-1} \otimes \mcW$, and thus the list of minimal energies $\{h_0, \ldots, h_{d-1}\}$ is the same for $\mcV$ and $\mcV \otimes \mcW$.
However, there is a constraint placed on the total $\sum h_i$ by the central charge which has been used in physics for classification of conformal field theories \cite{MathurMukhiSen88} (a mathematical proof of which may be found in \cite[\S 3]{Mason07}):

\begin{theorem}\label{thmExtremalityInequality}
Let $\mcV$ be a (rational, $C_2$-cofinite, CFT type) VOA with central charge $c$, and with lowest energies of modules given by $h_0, \ldots, h_{d-1}$.
Let $S \subseteq \{0, \ldots, d-1\}$ index a basis $\{\ch M_i\}_{i \in S}$ for $\Span \{\ch M_i\}$, and let $p = \abs{S}$.
Then
\begin{equation}\label{eqnMinimalEnergyInequality}
6\sum_{i \in S} h_i  \le \binom{p}{2} + \frac{pc}{4}.
\end{equation}
Moreover,
$$
\ell := \binom{p}{2} + \frac{p c}{4} - 6\sum_{i \in S} h_i
$$
is a non-negative integer.
\end{theorem}

As a consequence, tensoring by a holomorphic VOAs increases the parameter $\ell$, while leaving the representation category and minimal energies unchanged.
Thus one way of narrowing our field of study to reduce the presence of holomorphic VOAs is to consider only VOAs for which the sum $\sum h_i$ is as large as possible (that is, $\ell$ is as small as possible) for its central charge.

\begin{definition}\label{defExtremal}
Let $\mcV$ be a (rational, $C_2$-cofinite, CFT type) non-holomorphic VOA with central charge $c$.
Let $h_0, \ldots, h_{d-1}$ be the minimal energies for the irreducible modules of $\mcV$.
Assume that $c > 0$ and that $h_i > 0$ for $i > 0$.
Let $S \subseteq \{0, \ldots, d-1\}$ index a basis for $\Span \{\ch M_i\}$, and let $p = \abs{S}$.
Then $\mcV$ is called \emph{extremal} if
$$
\binom{p}{2} + \frac{pc}{4} < 6 \left( 1 + \sum_{i \in S} h_i \right).
$$
\end{definition}
It is clear that this definition is independent of the index set $S$ chosen.
Since the minimal energies of a VOA are determined, mod $1$, by $\Rep(\mcV)$, if $\mcV$ is an extremal VOA with central charge $c$  the sum of its minimal energies $\sum h_i$ is maximal among VOAs in the genus $(\Rep(\mcV), c)$.

\begin{remark}
The term \emph{extremal} was introduced in \cite{Hoehn95} to describe holomorphic VOAs with as large a gap as possible between the energy of the vacuum state and the next lowest energy state, and H\"ohn showed that the character of such a VOA is uniquely determined by the genus.
Our definition of extremal for non-holomorphic VOAs is similar in nature, in that certain lowest energies are required to be as large as possible.
We will observe in many examples that the characters of extremal non-holomorphic VOAs are fixed by their genus and minimal energies $h_i$.
\end{remark}

We can now state a pair of more modest conjectures.
\begin{conjecture}\label{cnjExtremalConjecture}
Let $\mcC$ be a non-trivial modular tensor category and $(\mcC, c)$ an admissible genus.
\begin{itemize}
\item (VOA version) There are finitely many isomorphism classes of extremal VOAs realizing $(\mcC, c)$.
\item (Character version) There are finitely many character vectors realized by extremal VOAs of genus $(\mcC, c)$.
\end{itemize}
\end{conjecture}
It is possible to have infinitely many extremal VOAs with the same representation category, for example $B_{a + 8k}$ for $k \in \{0, 1, 2, \ldots\}$.
It is also possible for there to be no extremal VOAs in a given genus, for example $(\Rep(\SU(2)_1), 25)$ and many other examples in Sections \ref{secRankTwo}-\ref{secRankThree}.
The character version of Conjecture \ref{cnjExtremalConjecture} appears quite a bit more tractable than Conjecture \ref{cnjCharacterFiniteness}, its analog without the extremality hypothesis.
In Theorem \ref{thmExtremalConjectureVerified}, we verify the character version of Conjecture \ref{cnjExtremalConjecture} when $\rank(\mcC) \le 3$, and compute the possible character vectors for $c \le 72$ (rank 2) and $c \le 48$ (rank 3). 

\subsection{Computing characters}\label{secComputingCharacters}

We will compute characters of extremal VOAs using a slightly modified version of the method of Bantay and Gannon.
The following brief introduction to vector-valued modular functions and the Bantay-Gannon method of computing fundamental matrices is adapted from \cite{BG,Gannon14}, and the interested reader may consult those sources for more details.

Let $\rho:\PSL(2,\Z) \to \U(d)$ be an irreducible representation, and suppose that $\bbX: \bbH \to \C^d$ is a holomorphic function defined on the upper half-plane $\bbH$ which satisfies
\begin{equation}\label{eqModularFunction2}
\bbX(\gamma \cdot \tau) = \rho(\gamma) \bbX(\tau)
\end{equation}
for all $\gamma \in \PSL(2,\Z)$ and $\tau \in \bbH$.
If $\Lambda$ satisfies $e^{2 \pi i \Lambda} = \rho \begin{pmatrix}1 & 1\\0 & 1\end{pmatrix}$, then $q^{-\Lambda}\bbX$ is invariant under the transformation $\tau \mapsto \tau + 1$, and thus we may Fourier expand
\begin{equation}\label{eqModularFunctionExpansion}
q^{-\Lambda} \bbX(\tau) = \sum_{n \in \Z} \bbX[n] q^n
\end{equation}
for some scalars $\bbX[n]$, where $q = e^{2 \pi i \tau}$.
We denote by $\mcM(\rho)$ the space of all holomorphic functions $\bbX: \bbH \to \C^d$ which satisfy condition \eqref{eqModularFunction2} and for which $\bbX[n] = 0$ for $n \ll 0$.

The matrix $\Lambda$ is called an \emph{exponent matrix} for $\rho$. 
Given a choice of exponent matrix, we define the \emph{principal part} map
$$
\mcP:\mcM(\rho) \longrightarrow \Span \{ vq^{-n}  \,\,\, | \,\,\, n > 0, \,v \in \C^{d}\}
$$
by
$$
\mcP\bbX(\tau) = \sum_{n < 0} \bbX[n] q^n.
$$
An exponent matrix is called \emph{bijective} if the principal part map is an isomorphism.

Bijective exponents exist for all $\rho$ \cite[Thm. 3.2]{Gannon14}, and a necessary condition is
\begin{equation}\label{eqTraceCondition}
\Tr(\Lambda) = \frac{5d}{12} + \frac{1}{4} \Tr(\rho\begin{pmatrix}0 & -1\\1&0\end{pmatrix}) + \frac{2}{3 \sqrt{3}} \Re \Big( e^{\frac{-i \pi}{6}} \Tr(\rho\begin{pmatrix}0 & -1\\ 1 & -1\end{pmatrix} \Big).
\end{equation}
When $d < 6$, it turns out that the trace condition \eqref{eqTraceCondition} is also sufficient for $\Lambda$ to be bijective \cite[Thm. 4.1]{Gannon14}.

Given a choice of bijective exponent $\Lambda$, we can define the \emph{canonical basis} vector $\bbX^{(\xi;n)} \in \mcM(\rho)$ to be the unique vector valued modular function satisfying 
$$
\mcP \bbX^{(\xi;n)}(q) = q^{-n} e_\xi,
$$
where $e_\xi$ is the standard basis vector in $\C^d$ with a $1$ in the position $\xi$, and zeros elsewhere.
The canonical basis vectors give a basis for $\mcM(\rho)$ as a vector space over $\C$, but of greater interest is the fact that the basis vectors with a simple pole
$$
\bbX^{(1;1)}, \ldots, \bbX^{(d;1)}
$$
are a basis for $\mcM(\rho)$ as a free $\C[J]$-module, where 
$$
J(\tau) = q^{-1} + 196884q + \cdots
$$
is the Klein $J$-invariant. 
The fundamental matrix $\Xi(\tau)$ is defined by
$$
\Xi = \left( \begin{array}{c|c|c} \bbX^{(1;1)} & \,\cdots & \bbX^{(d;1)} \end{array} \right).
$$

Let
$$
\frj(\tau) = \frac{984-J(\tau)}{1728},
$$
and regard $\Xi$ as a multivalued function of $\frj$.
Then the fundamental matrix satisfies the hypergeometric equation
$$
\frac{d \Xi(\frj)}{d \frj} = \Xi(j) \left( \frac{\mcA}{2 \frj} + \frac{\mcB}{3(\frj-1)} \right),
$$
where $\mcA$ and $\mcB$ can be written explicitly in terms of $\Lambda$ and the \emph{characteristic matrix}
$$
\chi := \left( \begin{array}{c|c|c} \bbX^{(1;1)}[0] & \,\cdots & \bbX^{(d;1)}[0] \end{array} \right)
$$
as follows:
\begin{align}\label{eqAandChi}
\mcA &= \frac{31}{36} (1 - \Lambda) - \frac{1}{864}(\chi + [\Lambda, \chi]),\\
\mcB &= \frac{41}{24} (1 - \Lambda) + \frac{1}{576}(\chi + [\Lambda, \chi]). \nonumber
\end{align}
By analyzing the spectra of $\mcA$ and $\mcB$, Bantay and Gannon showed that $\mcA$ satisfies the cubic equation
\begin{equation}\label{eqACubic}
\mcA \Lambda \mcA = -\frac{17}{18} \mcA - 2(\mcA \Lambda^2  + \Lambda \mcA \Lambda  + \Lambda^2 \mcA)  + 3(\mcA \Lambda + \Lambda \mcA) - 4\Lambda^3 + 8 \Lambda^2 - \frac{44}{9} \Lambda + \frac{8}{9}.
\end{equation}

One then proceeds by solving \eqref{eqACubic} for $\mcA$, up to the ambiguity of conjugating $\mcA$ by a (constant) matrix commuting with $\Lambda$.
Next, one solves the linear equation \eqref{eqAandChi} for $\chi$, still up to the same ambiguity.
The higher order coefficients of the fundamental matrix may then be calculated from a linear recurrence derived from the differential equation
\begin{equation}\label{eqnRecurrence}
\frac{1}{2 \pi i} \frac{d \Xi(\tau)}{d \tau} = \Xi(\tau) \mcD(\tau),
\end{equation}
where
$$
\mcD(\tau) = \frac{\Delta(\tau)}{E_{10}(\tau)} \Big( (J(\tau) -24)(\Lambda - 1) + \chi + [\Lambda, \chi] \Big).
$$
Here, $\Delta$ is the discriminant form of weight 12 and $E_{10}$ is the normalized Eisenstein series of weight 10.

At this point, one can recursively calculate the $q$ coefficients of $\Xi$ to high order, but still up to the ambiguity of conjugation by a matrix commuting with $\Lambda$.
In the examples we will look at, $\Lambda$ has distinct diagonal entries, and so $\Xi$ is determined up to conjugation by a diagonal matrix. 

The paper of Bantay and Gannon does not give a method to resolve this ambiguity, although in many examples they compute $\chi$ by other means.
This method was also used by Junla \cite{Junla}, who used ad hoc methods to resolve the ambiguity and compute characters in many examples with small central charge.
We use a more general approach.
Observe that the modular covariance condition \eqref{eqModularFunction} can hold for at most one element of the orbit of $\Xi$ under conjugation by a diagonal matrix.
Numerically solving \eqref{eqModularFunction} at fixed values of $\tau$ \footnote{For example, solving $\Xi(i) = \rho\begin{pmatrix} 0 & -1\\ 1 & 0 \end{pmatrix}\Xi(i)$.} then eliminates the ambiguity in $\Xi$.
Note that a numerical solution is generally sufficient for our purposes, since we will primarily be interested in whether the coefficients of $q$ in the first column of $\Xi$ are non-negative integers.

Now let $\mcV$ be a (rational, $C_2$-cofinite, CFT type) VOA with central charge $c$, and let $\rho$ be the representation of $\SL(2,\Z)$ associated to $\Rep(\mcV)$. 
Assume for the present that all modules $M_i$ of $\mcV$ are isomorphic to their dual, so that $\rho$ is in fact a representation of $\PSL(2,\Z)$.\footnote{This assumption is not essential, but useful to simplify exposition. To handle the non-self dual case, one can use the modification to the method of Bantay-Gannon discussed in \cite[Appendix A]{BG}}
Then by Zhu's theorem, 
$$
\begin{pmatrix} \ch M_0 \\ \vdots \\ \ch M_{d-1} \end{pmatrix} \in \mcM(\rho).
$$
Of particular interest to us is the observation that if the exponent matrix $\Lambda$ defined in terms of the minimal energies $h_i$ by $\Lambda_{ii} = \delta_{i0} + h_i - c/24$ is bijective, then the first column $\bbX^{(1;1)}$ of the fundamental matrix is the character vector of $\mcV$.
Moreover, since $\Lambda$ is bijective there are no other character vectors of VOAs in the genus $(\Rep(\mcV), c)$ with minimal energies $h_i^\prime \ge h_i$ for all $i$.

The contrapositive of the above observation is also of interest.
Suppose that $(\mcC, c)$ is an admissible genus (with all objects self-dual), and $\rho$ is the assoicated representation of $\PSL(2,\Z)$, with bijective exponent $\Lambda$ satisfying $\Lambda_{00} = 1$ and fundamental matrix $\Xi$.
Then if the coefficients of the $q$-expansion of the first column of $\Xi$ are not positive integers, then there is no VOA realizing $(\mcC, c)$ with minimal energies of non-vacuum modules satisfying $h_i \ge \Lambda_{ii} + c/24$.

We will use this fact to study the (non-)existance of VOAs realizing certain genera, with certain minimal energies $h_i$.
That is, given an admissible genus $(\mcC, c)$, we will look for all bijective exponent matrices $\Lambda$ with $\Lambda_{00} = 1$, compute the corresponding fundamental matrices $\Xi$, and check whether the coefficients of the first column are non-negative integers.

\section{Extremal characters in small genus}

\subsection{Main result}

The main result of this section gives a solution to Conjecture \ref{cnjExtremalConjecture} when $\rank(\mcC) \le 3$, and we list the potential character vectors for sufficiently small central charge.

\begin{theorem}\label{thmExtremalConjectureVerified}
Let $\mcC$ be a unitary modular tensor category, and let $(\mcC, c)$ be an admissible genus with $2 \le \rank(\mcC) \le 3$.
Then the character vector of an extremal VOA $\mcV$ with genus $(\mcC, c)$ is uniquely determined by the minimal energies $h_i$ of its modules.
In particular, the character version of Conjecture \ref{cnjExtremalConjecture} holds for these genera.

If $\rank(\mcC) = 2$ and $c \le 72$, then the characters of all extremal VOAs in the genus $(\mcC, c)$ are given in Section \ref{secRankTwo}.
If $\rank(\mcC) = 3$ and $c \le 48$, then the characters of all extremal VOAs in the genus $(\mcC, c)$ are given in Section \ref{secRankThree}.
\end{theorem}
\begin{proof}
It suffices to prove that if $\mcV$ is an extremal VOA in the genus $(\mcC, c)$, then $\Lambda_{ii} := h_i - c/24 + \delta_{i,0}$ defines a bijective exponent. 
We assume first that the objects of $\mcC$ are self-dual, in which case the extremality condition on $\mcV$ is equivalent to
\begin{equation}\label{eqnThmProofExtremality}
1 + \sum_{i =1}^{d-1} h_i  > \frac{1}{6}\binom{d}{2} + \frac{dc}{24} \ge \sum_{i=1}^{d-1} h_i,
\end{equation}
where $d = \rank(\mcC)$.
We used the fact that the characters of $\mcV$ must be linearly independent, since the simple objects of $\mcC$ have distinct twists.
From \eqref{eqnThmProofExtremality}, we can see that the total $\sum h_i$ for an extremal VOA in the genus $(\mcC, c)$ increases by $d$ when $c$ increases by $24$.

Since $\rank(\mcC) < 6$, to prove that $\Lambda$ is bijective it suffices by \cite[Thm. 4.1]{Gannon14} to verify that 
$$
\Tr(\Lambda) = \frac{5d}{12} + \frac{1}{4} \Tr(\rho\begin{pmatrix}0 & -1\\1&0\end{pmatrix}) + \frac{2}{3 \sqrt{3}} \Re \Big( e^{\frac{-i \pi}{6}} \Tr(\rho\begin{pmatrix}0 & -1\\ 1 & -1\end{pmatrix} \Big),
$$
or equivalently that
\begin{equation}\label{eqnThmProofBijective}
\sum_{i=1}^{d-1} h_i = -1 + \frac{dc}{24} + \frac{5d}{12} + \frac{1}{4} \Tr(\rho\begin{pmatrix}0 & -1\\1&0\end{pmatrix}) + \frac{2}{3 \sqrt{3}} \Re \Big( e^{\frac{-i \pi}{6}} \Tr(\rho\begin{pmatrix}0 & -1\\ 1 & -1\end{pmatrix} \Big).
\end{equation}
Regarding the total $\sum h_i$ for an extremal VOA as a function of $c$, we see that both sides of \eqref{eqnThmProofBijective} increase by $d$ when $c$ increases by $24$.
Hence it suffices to verify that $\Lambda$ is bijective for the first three admissible values of $c$.
This is easily done for each $\mcC$ with $\rank(\mcC) \le 3$ and all objects self-dual.

If the objects of $\mcC$ are not self-dual (i.e. if $\mcC$ has $\Z/3\Z$ fusion rules), then one can apply the argument of \cite[Appendix A]{BG} to reduce to the above case.

The tables in Sections \ref{secRankTwo}-\ref{secRankThree} are obtained by applying the modified version of the Bantay-Gannon method presented in Section \ref{secComputingCharacters} to compute the fundamental matrix $\Xi$ for the bijective exponent $\Lambda$ defined above, for all of the relevant genera (classified in \cite{RSW}) and all of the possible extremal choices of $h_i$ (modifying as in \cite[Appendix A]{BG} when not all objects are self-dual).
If there is an extremal VOA in the given genus with the given choice of $h_i$, its character vector must appear as the first column of one of these fundamental matrices, and so if the coefficients of its $q$-expansion are not all positive integers, there cannot be an extremal VOA corresponding to that choice of $h_i$.
On the other hand, if the first 100 coefficients of the $q$-expansion are all positive integers, we list this candidate character vector in the tables of Sections \ref{secRankTwo}-\ref{secRankThree}.
\end{proof}

In light of the proof of Theorem \ref{thmExtremalConjectureVerified}, it is natural to ask:

\begin{question}\label{quesTraceCondition}
Does the identity \eqref{eqnThmProofBijective} always hold when $\mcV$ is extremal and has $d$ linearly independent characters? Is the exponent $\Lambda_{ii} = h_i - c/24 + \delta_{i0}$ bijective?
\end{question}
Of course, one can formulate a version of Question \ref{quesTraceCondition} when the characters of $\mcV$ are not linearly independent as well.

\begin{remark}
If one were to extend our definition of extremality to holomorphic VOAs in the natural way (which would not generally agree with the definition in \cite{Hoehn95}), then the conclusion of Theorem \ref{thmExtremalConjectureVerified} holds when $\mcC$ is trivial as well, although we have not listed the characters in our tables.
In fact, the stronger VOA version of Conjecture \ref{cnjExtremalConjecture} holds in this case.
Indeed, the only holomorphic VOAs which satisfy our extremality condition are those with $c=8$ or $c=16$, and in both cases there are unique character vectors, the former with one realization, and the latter with two.
\end{remark}

The tables in Sections \ref{secRankTwo}-\ref{secRankThree} give all possible characters of extremal VOAs $\mcV$ with $\rank(\Rep(\mcV)) \in \{2, 3\}$, subject to the given bound on central charge, organized by $\Rep(\mcV)$.
All of the entries correspond to vector-valued modular functions for the appropriate representation of $\SL(2,\Z)$, and the first hundred coefficients of the $q$-expansions are all positive integers.
For each entry, we give the central charge and minimal energies that a realization would need to possess. 
While only the first few coefficients of the $q$-expansions are given, it is easy to calculate many more from the recurrence \eqref{eqnRecurrence}. 
We also highlight in each entry the dimension of the Lie algebra $\mcV_1$, and indicate whether we are aware of a VOA realization of the given candidate character vector.

All WZW models $\mcV_{\g,k}$ with $\g$ simple and $\rank(\Rep(\mcV_{\g,k})) \in \{2,3\}$ are extremal\footnote{i.e. $A_{1,2}$, $A_{2,1}$, $B_{n,1}$, $G_{2,1}$, $F_{4,1}$, $E_{6,1}$, $E_{7,1}$, $E_{8,2}$}, and all of their characters appear in the tables below, with the exception of $B_{n,1}$ which have too large of a central charge when $n$ is large.
The Ising minimal model also appears, as well as the baby monster VOA $V\!B_{(0)}^\natural$.
There are also several examples constructed as cosets in \cite{GHM}, based on predictions from \cite{HM}.

Many entries from our tables have not been realized.
It is possible that using the characters and representation categories, several of these can be realized via cosets or simple current extensions with a little effort (for example, we have done this with a candidate character vector with central charge $\frac{64}{7}$).
The (as-yet) unrealized entries are perhaps the simplest candidates for non-holomorphic VOAs which have not been constructed.

Not every vector valued modular form with positive integer coefficients corresponds to a VOA. 
Examples of spurious `candidate' character vectors can easily be constructed by taking linear combinations of characters of VOAs (in the genus $(\Rep(E_{7,1}), 23)$ one can find over $400$  `candidates' of this kind).
On the other hand, for our \emph{extremal} candidate characters listed in Sections \ref{secRankTwo}-\ref{secRankThree}, the extremality condition impiles that they cannot be constructed in this way; if one of our candidates in the tables does not correspond to a VOA, it would require an alternate explanation.

We close this section with a sample of unrealized candidate character vectors taken from Sections \ref{secRankTwo}-\ref{secRankThree}, along with the genus a realization would have.
$$
\begin{array}{|c|c|c|}
\hline
\mcC & c & \mbox{Candidate character vector}\\
\hline
\Rep(SU(2)_1) & 33 & 
q^{-33/24}\left(
\begin{array}{l}
1 + 3q + 86004q^2 + \cdots\\
q^{\frac{9}{4}}(565760 + 192053760q + \cdots)
\end{array}
\right)\\
\hline
\Ising & \frac{33}{2} &
q^{-33/48}\left(
\begin{array}{l}
1 + 231q + 38940q^2 + \cdots\\
q^{\frac{17}{16}}(528 + 70288q +  2186448q^2 + \cdots)\\
q^{\frac{3}{2}}(4301 + 247962q + 5625708q^2 + \cdots)
\end{array}
\right)\\
\hline
\tfrac{1}{2}\Rep(\SU(2)_5) & \frac{48}{7} & q^{-2/7}\left(
\begin{array}{l}
1 + 78q + 784q^2 + \cdots\\
q^{\frac{1}{7}}(1 + 133q + 1618q^2 + \cdots)\\
q^{\frac{5}{7}}(55 + 890q + 6720q^2 + \cdots)
\end{array}
\right)\\
\hline
\Rep(\SU(3)_1) & 34 &
q^{-17/12}\left(
\begin{array}{l}
1 + q + 58997q^2 + \cdots\\
q^{\frac{7}{3}}(1535274 + 528134256q +  \cdots)\\
q^{\frac{7}{3}}(1535274 + 528134256q +  \cdots)\\
\end{array}
\right)\\
\hline
\end{array}
$$

\subsection{Rank $=2$ extremal character candidates with $c \le 72$}\label{secRankTwo}

\subsubsection{$\SU(2)_1$ fusion rules}
\mbox{}\\

\noindent
$
\mcC = \Rep(SU(2)_1), \quad S = \frac{1}{\sqrt{2}}\begin{pmatrix} 1 & 1\\ 1 & -1 \end{pmatrix}, \quad \theta_1 = i
$
$$
\begin{array}{|c|c|c|c|c|c|}
\hline
c & h_1 & \ell & \dim \mcV_1 & \mbox{Realization?} & \mbox{Character vector}\\
\hline
1 & \frac{1}{4} &  0 & 3 & SU(2)_1 &
q^{-1/24}\left(
\begin{array}{l}
1 + 3q + 4q^2 + \cdots\\
q^{\frac{1}{4}}(2 + 2q + 6q^2 + \cdots)
\end{array}
\right)\\
\hline
9 & \frac{1}{4} &  4 & 251 & SU(2)_1 \otimes E_{8,1} &
q^{-9/24}\left(
\begin{array}{l}
1 + 251q + 4872q^2 + \cdots\\
q^{\frac{1}{4}}(2 + 498q + 8750q^2 + \cdots)
\end{array}
\right)\\
\hline
17 & \frac{5}{4} &  2 & 323 & \mbox{\cite{GHM}} &
q^{-17/24}\left(
\begin{array}{l}
1 + 323q + 60860q^2 + \cdots\\
q^{\frac{5}{4}}(1632 + 162656q + 4681120q^2 + \cdots)
\end{array}
\right)\\
\hline
33 & \frac{9}{4} &  4 & 3 & ? &
q^{-33/24}\left(
\begin{array}{l}
1 + 3q + 86004q^2 + \cdots\\
q^{\frac{9}{4}}(565760 + 192053760q + \cdots)
\end{array}
\right)\\
\hline
\end{array}
$$

\bigskip


\noindent
$
\mcC = \Rep(E_{7,1}), \quad S = \frac{1}{\sqrt{2}}\begin{pmatrix} 1 & 1\\ 1 & -1 \end{pmatrix}, \quad \theta_1 = -i
$

$$
\begin{array}{|c|c|c|c|c|c|}
\hline
c & h_1 & \ell & \dim \mcV_1 & \mbox{Realization?} & \mbox{Character vector}\\
\hline
7 & \frac{3}{4} &  0 & 133 & E_{7,1} &
q^{-7/24}\left(
\begin{array}{l}
1 + 133q + 1673q^2 + \cdots\\
q^{\frac{3}{4}}(56 + 968q + 7504q^2 + \cdots)
\end{array}
\right)\\
\hline
15 & \frac{3}{4} &  4 & 381 & E_{7,1} \otimes E_{8,1} &
q^{-15/24}\left(
\begin{array}{l}
1 + 381q + 38781q^2 + \cdots\\
q^{\frac{3}{4}}(56 + 14856q + 478512q^2 + \cdots)
\end{array}
\right)\\
\hline
23 & \frac{7}{4} &  2 & 69 & \mbox{\cite{GHM}} &
q^{-23/24}\left(
\begin{array}{l}
1 + 69q + 131905q^2 + \cdots\\
q^{\frac{7}{4}}(32384 + 4418944q + 189846784q^2 + \cdots)
\end{array}
\right)\\
\hline
\end{array}
$$

\bigskip

\subsubsection{Fibonacci fusion rules}


\mbox{}\\

\noindent
$
\mcC = \Rep(G_{2,1}), \quad \phi = \frac{1+\sqrt{5}}{2}, \quad S = \frac{1}{\sqrt{2+\phi}}\begin{pmatrix} 1 & \phi\\ \phi & -1 \end{pmatrix}, \quad \theta_1 = e^{4 \pi i/ 5}
$
$$
\begin{array}{|c|c|c|c|c|c|}
\hline
c & h_1 & \ell & \dim \mcV_1 & \mbox{Realization?} & \mbox{Character vector}\\
\hline
\frac{14}{5} & \frac{2}{5} &  0 & 14 & G_{2,1} &
q^{-7/60}\left(
\begin{array}{l}
1 + 14q + 42q^2 + \cdots\\
q^{\frac{2}{5}}(7 + 34q + 119q^2 + \cdots)
\end{array}
\right)\\
\hline
\frac{54}{5} & \frac{2}{5} &  4 & 262 & G_{2,1} \otimes E_{8,1} &
q^{-27/60}\left(
\begin{array}{l}
1 + 262q + 7638q^2 + \cdots\\
q^{\frac{2}{5}}(7 + 1770q + 37419q^2 + \cdots)
\end{array}
\right)\\
\hline
\frac{94}{5} & \frac{7}{5} &  2 & 188 & \mbox{\cite{GHM}} &
q^{-47/60}\left(
\begin{array}{l}
1 + 188q + 62087q^2 + \cdots\\
q^{\frac{7}{5}}(4794 + 532134q + 17518686q^2 + \cdots)
\end{array}
\right)\\
\hline
\end{array}
$$


\noindent
$
\mcC = \Rep(F_{4,1}), \quad \phi = \frac{1+\sqrt{5}}{2}, \quad S = \frac{1}{\sqrt{2+\phi}}\begin{pmatrix} 1 & \phi\\ \phi & -1 \end{pmatrix}, \quad \theta_1 = e^{6 \pi i/ 5}
$
$$
\begin{array}{|c|c|c|c|c|c|}
\hline
c & h_1 & \ell & \dim \mcV_1 & \mbox{Realization?} & \mbox{Character vector}\\
\hline
\frac{26}{5} & \frac{3}{5} &  0 & 52 & F_{4,1} &
q^{-13/60}\left(
\begin{array}{l}
1 + 52q + 377q^2 + \cdots\\
q^{\frac{3}{5}}(26 + 299q + 1702q^2 + \cdots)
\end{array}
\right)\\
\hline
\frac{66}{5} & \frac{3}{5} &  4 & 300 & F_{4,1} \otimes E_{8,1} &
q^{-33/60}\left(
\begin{array}{l}
1 + 300q + 17397q^2 + \cdots\\
q^{\frac{3}{5}}(26 + 6747q + 183078q^2 + \cdots)
\end{array}
\right)\\
\hline
\frac{106}{5} & \frac{8}{5} &  2 & 106 & \mbox{\cite{GHM}} &
q^{-33/60}\left(
\begin{array}{l}
1 + 106q + 84429q^2 +\cdots\\
q^{\frac{8}{5}}(15847 + 1991846q + 76895739q^2 + \cdots)
\end{array}
\right)\\
\hline
\end{array}
$$

\bigskip

\newpage

\subsection{Rank $=3$ extremal character candidates with $c \le 48$}\label{secRankThree}

\subsubsection{Ising fusion rules}
\mbox{}\\

\noindent
$
\mcC = \Ising,  \quad S = \frac{1}{2}\begin{pmatrix} 1 & \sqrt{2} & 1\\ \sqrt{2} & 0 & -\sqrt{2}\\ 1 & -\sqrt{2} & 1 \end{pmatrix}, \quad \theta_1 = e^{\pi i/ 8}, \quad \theta_2 = -1
$
$$
\begin{array}{|c|c|c|c|c|c|c|}
\hline
c & h_1 & h_2 & \ell & \dim \mcV_1 & \mbox{Realization?} & \mbox{Character vector}\\
\hline
\frac{1}{2} & \frac{1}{16} & \frac{1}{2} & 0 & 0 &\begin{array}{c}M(4,3)\\\mbox{\tiny (Ising)}\end{array} &
q^{-1/48}\left(
\begin{array}{l}
1 + 0q + q^2 + \cdots\\
q^{\frac{1}{16}}(1 + q + q^2 + \cdots)\\
q^{\frac{1}{2}}(1 + q + q^2 + \cdots)
\end{array}
\right)\\
\hline
\frac{17}{2} & \frac{17}{16} & \frac{1}{2} & 0 & 136 &B_{8,1} & 
q^{-17/48}\left(
\begin{array}{l}
1 + 136q + 2669q^2 + \cdots\\
q^{\frac{17}{16}}(256 + 4352q + 39168q^2 + \cdots)\\
q^{\frac{1}{2}}(17 + 697q + 8517q^2 + \cdots)
\end{array}
\right)\\
\hline
\frac{33}{2} & \frac{33}{16} & \frac{1}{2} & 0 & 528 &B_{16,1}&
q^{-33/48}\left(
\begin{array}{l}
1 + 528q + 42009q^2 + \cdots\\
q^{\frac{33}{16}}(65536 + 2162688q +  \cdots)\\
q^{\frac{1}{2}}(33 + 5489q + 254793q^2 + \cdots)
\end{array}
\right)\\
\hline
\frac{33}{2} & \frac{17}{16} & \frac{3}{2} & 0 & 231 &\mbox{?} &
q^{-33/48}\left(
\begin{array}{l}
1 + 231q + 38940q^2 + \cdots\\
q^{\frac{17}{16}}(528 + 70288q +  2186448q^2 + \cdots)\\
q^{\frac{3}{2}}(4301 + 247962q + 5625708q^2 + \cdots)
\end{array}
\right)\\
\hline
\frac{49}{2} & \frac{49}{16} & \frac{1}{2} & 0 & 1176 & B_{24,1} &
q^{-49/48}\left(
\begin{array}{l}
1 + 1176q + 214277q^2 + \cdots\\
q^{\frac{49}{16}}(16777216 + 822083584q +  \cdots)\\
q^{\frac{1}{2}}(49 + 18473q + 1964557q^2 + \cdots)
\end{array}
\right)\\
\hline
\end{array}
$$
We have omitted $B_{32,1}$ and $B_{40,1}$ at central charges $65/2$ and $81/2$, respectively.

\bigskip

\newpage


\noindent 
$
\mcC = \Rep(SU(2)_2),  \quad S = \frac{1}{2}\begin{pmatrix} 1 & \sqrt{2} & 1\\ \sqrt{2} & 0 & -\sqrt{2}\\ 1 & -\sqrt{2} & 1 \end{pmatrix}, \quad \theta_1 = e^{3\pi i/ 8}, \quad \theta_2 = -1
$
$$
\begin{array}{|c|c|c|c|c|c|c|}
\hline
c & h_1 & h_2 & \ell & \dim \mcV_1 & \mbox{Realization?} & \mbox{Character vector}\\
\hline
\frac{3}{2} & \frac{3}{16} & \frac{1}{2} & 0 & 3 & SU(2)_2 &
q^{-3/48}\left(
\begin{array}{l}
1 + 3q + 9q^2 + \cdots\\
q^{\frac{3}{16}}(2 + 6q + 12q^2 + \cdots)\\
q^{\frac{1}{2}}(3 + 4q + 12q^2 + \cdots)
\end{array}
\right)\\
\hline
\frac{19}{2} & \frac{19}{16} & \frac{1}{2} & 0 & 171 & B_{9,1} &
q^{-19/48}\left(
\begin{array}{l}
1 + 171q + 4237q^2 + \cdots\\
q^{\frac{19}{16}}(512 + 9728q + 97280q^2 + \cdots)\\
q^{\frac{1}{2}}(19 + 988q + 14896q^2 + \cdots)
\end{array}
\right)\\
\hline
\frac{35}{2} & \frac{35}{16} & \frac{1}{2} & 0 & 595 & B_{17,1} &
q^{-35/48}\left(
\begin{array}{l}
1 + 595q + 53585q^2 + \cdots\\
q^{\frac{35}{16}}(131072 + 4587520q +  \cdots)\\
q^{\frac{1}{2}}(35 + 6580q + 345492q^2 + \cdots)
\end{array}
\right)\\
\hline
\frac{35}{2} & \frac{19}{16} & \frac{3}{2} & 0 & 210 & \mbox{\cite{GHM}} &
q^{-35/48}\left(
\begin{array}{l}
1 + 210q + 47425q^2 + \cdots\\
q^{\frac{19}{16}}(1120 + 143392q +  4661440q^2 + \cdots)\\
q^{\frac{3}{2}}(4655 + 329707q + 8512950q^2 + \cdots)
\end{array}
\right)\\
\hline
\frac{51}{2} & \frac{51}{16} & \frac{1}{2} & 0 & 1275 & B_{25,1} &
q^{-51/48}\left(
\begin{array}{l}
1 + 1275q + 252501q^2 + \cdots\\
q^{\frac{51}{16}}(417792 + 44834816q +  \cdots)\\
q^{\frac{1}{2}}(2975 + 1481907q + \cdots)
\end{array}
\right)\\
\hline
\end{array}
$$
We have omitted $B_{33,1}$ and $B_{41,1}$ at central charges $c= 67/2$ and $c=83/2$, respectively.

\bigskip


\noindent 
$
\mcC = \Rep(B_{2,1}),  \quad S = \frac{1}{2}\begin{pmatrix} 1 & \sqrt{2} & 1\\ \sqrt{2} & 0 & -\sqrt{2}\\ 1 & -\sqrt{2} & 1 \end{pmatrix}, \quad \theta_1 = e^{5\pi i/ 8}, \quad \theta_2 = -1
$

$$
\begin{array}{|c|c|c|c|c|c|c|}
\hline
c & h_1 & h_2 & \ell & \dim \mcV_1 & \mbox{Realization?} & \mbox{Character vector}\\
\hline
\frac{5}{2} & \frac{5}{16} & \frac{1}{2} & 0 & 10 & B_{2,1} &
q^{-5/48}\left(
\begin{array}{l}
1 + 10q + 30q^2 + \cdots\\
q^{\frac{5}{16}}(4 + 20q + 60q^2 + \cdots)\\
q^{\frac{1}{2}}(5 + 15q + 56q^2 + \cdots)
\end{array}
\right)\\
\hline
\frac{21}{2} & \frac{21}{16} & \frac{1}{2} & 0 & 210 & B_{10,1} &
q^{-21/48}\left(
\begin{array}{l}
1 + 210q + 6426q^2 + \cdots\\
q^{\frac{21}{16}}(1024 + 21504q + 236544q^2 + \cdots)\\
q^{\frac{1}{2}}(21 + 1351q + 24780q^2 + \cdots)
\end{array}
\right)\\
\hline
\frac{37}{2} & \frac{21}{16} & \frac{3}{2} & 0 & 185 & \mbox{\cite{GHM}} &
q^{-37/48}\left(
\begin{array}{l}
1 + 185q + 56351q^2 + \cdots\\
q^{\frac{21}{16}}(2368 + 292928q + 9914816q^2 + \cdots)\\
q^{\frac{3}{2}}(4921 + 427868q + 12578261q^2 + \cdots)
\end{array}
\right)\\
\hline
\frac{37}{2} & \frac{37}{16} & \frac{1}{2} & 0 & 666 & B_{18,1} &
q^{-37/48}\left(
\begin{array}{l}
1 + 666q + 67414q^2 + \cdots\\
q^{\frac{37}{16}}(262144 + 9699328q + 184287232q^2 + \cdots)\\
q^{\frac{1}{2}}(37 + 7807q + 460576q^2 + \cdots)
\end{array}
\right)\\
\hline
\end{array}
$$
We have omitted $B_{26,1}$, $B_{34,1}$, and $B_{42,1}$ at central charges $c=53/2$, $c=69/2$, and $c=85/2$ respectively.

\bigskip


\noindent 
$
\mcC = \Rep(B_{3,1}),  \quad S = \frac{1}{2}\begin{pmatrix} 1 & \sqrt{2} & 1\\ \sqrt{2} & 0 & -\sqrt{2}\\ 1 & -\sqrt{2} & 1 \end{pmatrix}, \quad \theta_1 = e^{7\pi i/ 8}, \quad \theta_2 = -1
$

$$
\begin{array}{|c|c|c|c|c|c|c|}
\hline
c & h_1 & h_2 & \ell & \dim \mcV_1 & \mbox{Realization?} & \mbox{Character vector}\\
\hline
\frac{7}{2} & \frac{7}{16} & \frac{1}{2} & 0 & 21 & B_{3,1} &
q^{-7/48}\left(
\begin{array}{l}
1 + 21q + 84q^2 + \cdots\\
q^{\frac{7}{16}}(8 + 56q + 224q^2 + \cdots)\\
q^{\frac{1}{2}}(7 + 42q + 175q^2 + \cdots)
\end{array}
\right)\\
\hline
\frac{23}{2} & \frac{23}{16} & \frac{1}{2} & 0 & 253 & B_{11,1} &
q^{-23/48}\left(
\begin{array}{l}
1 + 253q + 9384q^2 + \cdots\\
q^{\frac{23}{16}}(2048 + 47104q + 565248q^2 + \cdots)\\
q^{\frac{1}{2}}(23 + 1794q + 39491q^2 + \cdots)
\end{array}
\right)\\
\hline
\frac{39}{2} & \frac{39}{16} & \frac{1}{2} & 0 & 741 & B_{19,1} &
q^{-39/48}\left(
\begin{array}{l}
1 + 741q + 83772q^2 + \cdots\\
q^{\frac{23}{16}}(524288 + 20447232q + 408944640q^2 + \cdots)\\
q^{\frac{1}{2}}(39 + 9178q + 604695q^2 + \cdots)
\end{array}
\right)\\
\hline
\frac{39}{2} & \frac{23}{16} & \frac{3}{2} & 0 & 156 & \mbox{\cite{GHM}} &
q^{-39/48}\left(
\begin{array}{l}
1 + 156q + 65442q^2 + \cdots\\
q^{\frac{23}{16}}(4992 + 599168q + 21046272q^2 + \cdots)\\
q^{\frac{3}{2}}(5083 + 542685q + 18172323q^2 + \cdots)
\end{array}
\right)\\
\hline
\end{array}
$$
We have omitted $B_{27,1}$, $B_{35,1}$ and $B_{43,1}$ at central charges $c=55/2$, $c=71/2$ and $c=87/2$, respectively.

\bigskip

\noindent 
$
\mcC = \Rep(B_{4,1}),  \quad S = \frac{1}{2}\begin{pmatrix} 1 & \sqrt{2} & 1\\ \sqrt{2} & 0 & -\sqrt{2}\\ 1 & -\sqrt{2} & 1 \end{pmatrix}, \quad \theta_1 = e^{9\pi i/ 8}, \quad \theta_2 = -1
$

$$
\begin{array}{|c|c|c|c|c|c|c|}
\hline
c & h_1 & h_2 & \ell & \dim \mcV_1 & \mbox{Realization?} & \mbox{Character vector}\\
\hline
\frac{9}{2} & \frac{9}{16} & \frac{1}{2} & 0 & 36 & B_{4,1} &
q^{-9/48}\left(
\begin{array}{l}
1 + 36q + 207q^2 + \cdots\\
q^{\frac{9}{16}}(16 + 144q + 720q^2 + \cdots)\\
q^{\frac{1}{2}}(9 + 93q + 459q^2 + \cdots)
\end{array}
\right)\\
\hline
\frac{25}{2} & \frac{25}{16} & \frac{1}{2} & 0 & 300 & B_{12,1} &
q^{-9/48}\left(
\begin{array}{l}
1 + 300q + 13275q^2 + \cdots\\
q^{\frac{25}{16}}(4096 + 102400q + 1331200q^2 + \cdots)\\
q^{\frac{1}{2}}(25 + 2325q + 60655q^2 + \cdots)
\end{array}
\right)\\
\hline
\frac{25}{2} & \frac{9}{16} & \frac{3}{2} & 0 & 275 & ? &
q^{-25/48}\left(
\begin{array}{l}
1 + 275q + 13250q^2 + \cdots\\
q^{\frac{9}{16}}(25 + 4121q + 102425q^2 + \cdots)\\
q^{\frac{3}{2}}(2325 + 60630q + 811950q^2 + \cdots)
\end{array}
\right)\\
\hline
\frac{41}{2} & \frac{41}{16} & \frac{1}{2} & 0 & 820 & B_{20,1} &
q^{-41/48}\left(
\begin{array}{l}
1 + 820q + 102951q^2 + \cdots\\
q^{\frac{41}{16}}(1048576 + 42991616q +  \cdots)\\
q^{\frac{1}{2}}(41 + 10701q + 783059q^2 + \cdots)
\end{array}
\right)\\
\hline
\frac{41}{2} & \frac{25}{16} & \frac{3}{2} & 0 & 123 & \mbox{\cite{GHM}} &
q^{-41/48}\left(
\begin{array}{l}
1 + 123q + 74374q^2 + \cdots\\
q^{\frac{25}{16}}(10496 + 1227008q + 44597504q^2 + \cdots)\\
q^{\frac{3}{2}}(5125 + 673630q + 25702490q^2 + \cdots)
\end{array}
\right)\\
\hline
\end{array}
$$
We have omitted $B_{28,1}$, $B_{36,1}$, and $B_{44,1}$ at central charges $c=57/2$, $c=73/2$, and $c=89/2$, respectively.




\noindent 
$
\mcC = \Rep(B_{5,1}),  \quad S = \frac{1}{2}\begin{pmatrix} 1 & \sqrt{2} & 1\\ \sqrt{2} & 0 & -\sqrt{2}\\ 1 & -\sqrt{2} & 1 \end{pmatrix}, \quad \theta_1 = e^{11\pi i/ 8}, \quad \theta_2 = -1
$
$$
\begin{array}{|c|c|c|c|c|c|c|}
\hline
c & h_1 & h_2 & \ell & \dim \mcV_1 & \mbox{Realization?} & \mbox{Character vector}\\
\hline
\frac{11}{2} & \frac{11}{16} & \frac{1}{2} & 0 & 55 & B_{5,1} &
q^{-11/48}\left(
\begin{array}{l}
1 + 55q + 451q^2 + \cdots\\
q^{\frac{11}{16}}(32 + 352q + 2112q^2 + \cdots)\\
q^{\frac{1}{2}}(11 + 176q + 1078q^2 + \cdots)
\end{array}
\right)\\
\hline
\frac{27}{2} & \frac{27}{16} & \frac{1}{2} & 0 & 351 & B_{13,1} &
q^{-27/48}\left(
\begin{array}{l}
1 + 351q + 18279q^2 + \cdots\\
q^{\frac{27}{16}}(8192 + 221184q + 3096576q^2 + \cdots)\\
q^{\frac{1}{2}}(27 + 2952q + 90234q^2 + \cdots)
\end{array}
\right)\\
\hline
\frac{27}{2} & \frac{11}{16} & \frac{3}{2} & 0 & 270 & ? &
q^{-27/48}\left(
\begin{array}{l}
1 + 270q + 18171q^2 + \cdots\\
q^{\frac{11}{16}}(54 + 8354q + 221508q^2 + \cdots)\\
q^{\frac{3}{2}}(2871 + 89991q + 1380456q^2 + \cdots)
\end{array}
\right)\\
\hline
\frac{43}{2} & \frac{43}{16} & \frac{1}{2} & 0 & 903 & B_{21,1} &
q^{-43/48}\left(
\begin{array}{l}
1 + 903q + 125259q^2 + \cdots\\
q^{\frac{43}{16}}(2097152 + 90177536q +\cdots)\\
q^{\frac{1}{2}}(43 + 12384q + 1001470q^2 + \cdots)
\end{array}
\right)\\
\hline
\frac{43}{2} & \frac{27}{16} & \frac{3}{2} & 0 & 86 & \mbox{\cite{GHM}} &
q^{-43/48}\left(
\begin{array}{l}
1 + 86q + 82775q^2 + \cdots\\
q^{\frac{27}{16}}(22016 + 2515456q + 94360576q^2 + \cdots)\\
q^{\frac{3}{2}}(5031 + 819279q + 35627220q^2 + \cdots)
\end{array}
\right)\\
\hline
\end{array}
$$
We have omitted $B_{29,1}$, $B_{37,1}$, and $B_{45,1}$ at central charges $c=59/2$, $c=75/2$, and $c=91/2$, respectively.
\bigskip


\noindent
$
\mcC = \Rep(B_{6,1}),  \quad S = \frac{1}{2}\begin{pmatrix} 1 & \sqrt{2} & 1\\ \sqrt{2} & 0 & -\sqrt{2}\\ 1 & -\sqrt{2} & 1 \end{pmatrix}, \quad \theta_1 = e^{13\pi i/ 8}, \quad \theta_2 = -1
$

$$
\begin{array}{|c|c|c|c|c|c|c|}
\hline
c & h_1 & h_2 & \ell & \dim \mcV_1 & \mbox{Realization?} & \mbox{Character vector}\\
\hline
\frac{13}{2} & \frac{13}{16} & \frac{1}{2} & 0 & 78 & B_{6,1} &
q^{-13/48}\left(
\begin{array}{l}
1 + 78q + 884q^2 + \cdots\\
q^{\frac{13}{16}}(64 + 832q + 5824q^2 + \cdots)\\
q^{\frac{1}{2}}(13 + 299q + 2314q^2 + \cdots)
\end{array}
\right)\\
\hline
\frac{29}{2} & \frac{29}{16} & \frac{1}{2} & 0 & 406 & B_{14,1} &
q^{-29/48}\left(
\begin{array}{l}
1 + 406q + 24592q^2 + \cdots\\
q^{\frac{29}{16}}(16384 + 475136q + 7127040q^2 + \cdots)\\
q^{\frac{1}{2}}(29 + 3683q + 130558q^2 + \cdots)
\end{array}
\right)\\
\hline
\frac{29}{2} & \frac{13}{16} & \frac{3}{2} & 0 & 261 & ? &
q^{-29/48}\left(
\begin{array}{l}
1 + 261q + 24157q^2 + \cdots\\
q^{\frac{13}{16}}(116 + 16964q + 476876q^2 + \cdots)\\
q^{\frac{3}{2}}(3393 + 129688q + 2279671q^2 + \cdots)
\end{array}
\right)\\
\hline
\frac{45}{2} & \frac{45}{16} & \frac{1}{2} & 0 & 990 & B_{22,1} &
q^{-45/48}\left(
\begin{array}{l}
1 + 990q + 151020q^2 + \cdots\\
q^{\frac{45}{16}}(4194304 + 188743680q +  \cdots)\\
q^{\frac{1}{2}}(45 + 14235q + 1266354q^2 + \cdots)
\end{array}
\right)\\
\hline
\frac{45}{2} & \frac{29}{16} & \frac{3}{2} & 0 & 45 &  \mbox{\cite{GHM}} &
q^{-45/48}\left(
\begin{array}{l}
1 + 45q + 90225q^2 + \cdots\\
q^{\frac{29}{16}}(46080 + 5161984q +  \cdots)\\
q^{\frac{3}{2}}(4785 + 977184q + 48445515q^2 + \cdots)
\end{array}
\right)\\
\hline
\end{array}
$$
We have omitted $B_{30,1}$, $B_{38,1}$, and $B_{46,1}$ at central charges $c=61/2$, $c=77/2$, and $c=93/2$, respectively.

\bigskip


\noindent
$
\mcC = \Rep(B_{7,1}),  \quad S = \frac{1}{2}\begin{pmatrix} 1 & \sqrt{2} & 1\\ \sqrt{2} & 0 & -\sqrt{2}\\ 1 & -\sqrt{2} & 1 \end{pmatrix}, \quad \theta_1 = e^{15\pi i/ 8}, \quad \theta_2 = -1
$

$$
\begin{array}{|c|c|c|c|c|c|c|}
\hline
c & h_1 & h_2 & \ell & \dim \mcV_1 & \mbox{Realization?} & \mbox{Character vector}\\
\hline
\frac{15}{2} & \frac{15}{16} & \frac{1}{2} & 0 & 105 & B_{7,1} &
q^{-15/48}\left(
\begin{array}{l}
1 + 105q + 1590q^2 + \cdots\\
q^{\frac{15}{16}}(128 +1920 q + 15360q^2 + \cdots)\\
q^{\frac{1}{2}}(15 + 470q + 4593q^2 + \cdots)
\end{array}
\right)\\
\hline
\frac{31}{2} & \frac{31}{16} & \frac{1}{2} & 0 & 465 & B_{15,1} &
q^{-31/48}\left(
\begin{array}{l}
1 + 465q + 32426q^2 + \cdots\\
q^{\frac{31}{16}}(32768 +1015808 q + 16252928q^2 + \cdots)\\
q^{\frac{1}{2}}(31 + 4526q + 184357q^2 + \cdots)
\end{array}
\right)\\
\hline
\frac{31}{2} & \frac{15}{16} & \frac{3}{2} & 0 & 248 & E_{8,2} &
q^{-31/48}\left(
\begin{array}{l}
1 + 248q + 31124q^2 + \cdots\\
q^{\frac{15}{16}}(248 +34504 q + 1022752q^2 + \cdots)\\
q^{\frac{3}{2}}(3875 + 181753q + 3623869q^2 + \cdots)
\end{array}
\right)\\
\hline
\frac{47}{2} & \frac{47}{16} & \frac{1}{2} & 0 & 1081 & B_{23,1} &
q^{-47/48}\left(
\begin{array}{l}
1 + 1081q + 180574q^2 + \cdots\\
q^{\frac{47}{16}}(8388608 +394264576 q + \cdots)\\
q^{\frac{1}{2}}(47 + 16262q + 1584793q^2 + \cdots)
\end{array}
\right)\\
\hline
\frac{47}{2} & \frac{31}{16} & \frac{3}{2} & 0 & 0 & \begin{array}{c}{V\!B}_{(0)}^\natural\\\mbox{\tiny (Baby monster)}\end{array} &
q^{-47/48}\left(
\begin{array}{l}
1 + 0q + 96256q^2 + \cdots\\
q^{\frac{31}{16}}(96256 +10602496 q + \cdots)\\
q^{\frac{3}{2}}(4371 + 1143745q + 64680601q^2 + \cdots)
\end{array}
\right)\\
\hline
\end{array}
$$
We have omitted $B_{31,1}$, $B_{39,1}$, and $B_{47,1}$ at central charges $c=63/2$, $c=79/2$, and $c=95/2$, respectively.

\bigskip

\newpage

\subsubsection{$\tfrac{1}{2}\SU(2)_5$ fusion rules}

\mbox{}\\

\noindent
$
\mcC = \tfrac{1}{2} \Rep(\SU(2)_5),  \quad \psi = 2\cos(\tfrac{\pi}{7}), \quad S = \frac{2 \sin(\pi/7)}{\sqrt{7}}\begin{pmatrix} 1 & \psi & \psi^2\! -\! 1\\ \psi & 1\!-\!\psi^2 & 1\\ \psi^2 \!-\! 1 & 1 & -\psi \end{pmatrix},$

$\theta_1 = e^{2\pi i/ 7}, \quad \theta_2 = e^{10\pi i/7}
$
$$
\begin{array}{|c|c|c|c|c|c|c|}
\hline
c & h_1 & h_2 & \ell & \dim \mcV_1 & \mbox{Realization?} & \mbox{Character vector}\\
\hline
\frac{48}{7} & \frac{1}{7} & \frac{5}{7} & 3 & 78 &\mbox{?} &
q^{-2/7}\left(
\begin{array}{l}
1 + 78q + 784q^2 + \cdots\\
q^{\frac{1}{7}}(1 + 133q + 1618q^2 + \cdots)\\
q^{\frac{5}{7}}(55 + 890q + 6720q^2 + \cdots)
\end{array}
\right)\\
\hline
\frac{104}{7} & \frac{8}{7} & \frac{5}{7} & 3 & 188 &\mbox{?} &
q^{-13/21}\left(
\begin{array}{l}
1 + 188q + 17260q^2 + \cdots\\
q^{\frac{8}{7}}(725 + 52316q + 1197468q^2 + \cdots)\\
q^{\frac{5}{7}}(44 + 13002q + 424040q^2 + \cdots)
\end{array}
\right)\\
\hline
\frac{160}{7} & \frac{8}{7} & \frac{12}{7} & 3 & 40 &\mbox{?} &
q^{-20/21}\left(
\begin{array}{l}
1 + 40q + 60440q^2 + \cdots\\
q^{\frac{8}{7}}(285 + 227848q + 17128120q^2 + \cdots)\\
q^{\frac{12}{7}}(27170 + 3857360q + \cdots)
\end{array}
\right)\\
\hline
\frac{216}{7} & \frac{15}{7} & \frac{12}{7} & 3 & 3 &\mbox{?} &
q^{-9/7}\left(
\begin{array}{l}
1 + 3q + 52254q^2 + \cdots\\
q^{\frac{15}{7}}(260623 + 74348634q + \cdots)\\
q^{\frac{12}{7}}(11495 + 10341870q + \cdots)
\end{array}
\right)\\
\hline
\end{array}
$$

\bigskip


\noindent
$
\mcC = \overline{\tfrac{1}{2} \Rep(\SU(2)_5)},  \quad \psi = 2\cos(\tfrac{\pi}{7}), \quad S = \frac{2 \sin(\frac{\pi}{7})}{\sqrt{7}}\begin{pmatrix} 1 & \psi & \psi^2\! -\! 1\\ \psi & 1\!-\!\psi^2 & 1\\ \psi^2 \!-\! 1 & 1 & -\psi \end{pmatrix},$

$\theta_1 = e^{12\pi i/ 7}, \quad \theta_2 = e^{4\pi i/7}
$
$$
\begin{array}{|c|c|c|c|c|c|c|}
\hline
c & h_1 & h_2 & \ell & \dim \mcV_1 & \mbox{Realization?} & \mbox{Character vector}\\
\hline
\frac{64}{7} & \frac{6}{7} & \frac{2}{7} & 3 & 136 & \mbox{\small Yes} &
q^{-8/21}\left(
\begin{array}{l}
1 + 136q + 2417q^2 + \cdots\\
q^{\frac{6}{7}}(117 + 2952q + 32220q^2 + \cdots)\\
q^{\frac{2}{7}}(3 + 632q + 10787q^2 + \cdots)
\end{array}
\right)\\
\hline
\frac{120}{7} & \frac{6}{7} & \frac{9}{7} & 3 & 156 &\mbox{?} &
q^{-15/21}\left(
\begin{array}{l}
1 + 156q + 28926q^2 + \cdots\\
q^{\frac{6}{7}}(78 + 28692q + 1194804q^2 + \cdots)\\
q^{\frac{9}{7}}(2108 + 200787q + 5744052q^2 + \cdots)
\end{array}
\right)\\
\hline
\frac{176}{7} & \frac{13}{7} & \frac{9}{7} & 3 & 14 &\mbox{?} &
q^{-22/21}\left(
\begin{array}{l}
1 + 14q + 66512q^2 + \cdots\\
q^{\frac{13}{7}}(50922 + 8656740q + \cdots)\\
q^{\frac{9}{7}}(782 + 718267q + 64206178q^2 + \cdots)
\end{array}
\right)\\
\hline
\end{array}
$$

We have verified at the level of characters that the example with $c=64/7$ is realized as a simple current extension of $\SU(2)_5 \otimes E_{7,1}$.

\newpage

\subsubsection{$\SU(3)_1$ fusion rules}
\mbox{}\\

\noindent
$
\mcC = \Rep(\SU(3)_1),  \quad \omega = e^{2\pi i/3}, \quad S = \frac{1}{\sqrt{3} }\begin{pmatrix} 1 & 1 & 1\\ 1 & \omega & \omega^2\\ 1 & \omega^2 & \omega \end{pmatrix}, \quad \theta_1 = \theta_2 = \omega
$
$$
\begin{array}{|c|c|c|c|c|c|c|}
\hline
c & h_1 & h_2 & \ell & \dim \mcV_1 & \mbox{Realization?} & \mbox{Character vector}\\
\hline
2 & \frac{1}{3} & \frac{1}{3} & 0 & 8 & \SU(3)_1 &
q^{-1/12}\left(
\begin{array}{l}
1 + 8q + 17q^2 + \cdots\\
q^{\frac{1}{3}}(6 + 18q + 54q^2 + \cdots)\\
q^{\frac{1}{3}}(6 + 18q + 54q^2 + \cdots)
\end{array}
\right)\\
\hline
10 & \frac{1}{3} & \frac{1}{3} & 4 & 256 & \SU(3)_1 \otimes E_{8,1} &
q^{-5/12}\left(
\begin{array}{l}
1 + 256q + 6125q^2 + \cdots\\
q^{\frac{1}{3}}(6 + 1506q + 29262q^2 + \cdots)\\
q^{\frac{1}{3}}(6 + 1506q + 29262q^2 + \cdots)
\end{array}
\right)\\
\hline
18 & \frac{4}{3} & \frac{4}{3} & 2 & 234 & \mbox{\cite{GHM}} &
q^{-9/12}\left(
\begin{array}{l}
1 + 234q + 59805q^2 + \cdots\\
q^{\frac{4}{3}}(4374 + 463644q + 14403582q^2 + \cdots)\\
q^{\frac{4}{3}}(4374 + 463644q + 14403582q^2 + \cdots)
\end{array}
\right)\\
\hline
34 & \frac{7}{3} & \frac{7}{3} & 4 & 1 & ? &
q^{-17/12}\left(
\begin{array}{l}
1 + q + 58997q^2 + \cdots\\
q^{\frac{7}{3}}(1535274 + 528134256q +  \cdots)\\
q^{\frac{7}{3}}(1535274 + 528134256q +  \cdots)\\
\end{array}
\right)\\
\hline
\end{array}
$$

\bigskip


\noindent
$
\mcC = \Rep(E_{6,1}),  \quad \omega = e^{2\pi i/3}, \quad S = \frac{1}{\sqrt{3} }\begin{pmatrix} 1 & 1 & 1\\ 1 & \omega^2 & \omega\\ 1 & \omega & \omega^2 \end{pmatrix}, \quad \theta_1 = \theta_2 = \omega^2
$
$$
\begin{array}{|c|c|c|c|c|c|c|}
\hline
c & h_1 & h_2 & \ell & \dim \mcV_1 & \mbox{Realization?} & \mbox{Character vector}\\
\hline
6 & \frac{2}{3} & \frac{2}{3} & 0 & 78 & E_{6,1} &
q^{-1/4}\left(
\begin{array}{l}
1 + 78q + 729q^2 + \cdots\\
q^{\frac{2}{3}}(54 + 756q + 4968q^2 + \cdots)\\
q^{\frac{2}{3}}(54 + 756q + 4968q^2 + \cdots)
\end{array}
\right)\\
\hline
14 & \frac{2}{3} & \frac{2}{3} & 4 & 326 & E_{6,1} \otimes E_{8,1} &
q^{-7/12}\left(
\begin{array}{l}
1 + 326q + 24197q^2 + \cdots\\
q^{\frac{2}{3}}(54 + 14148q + 415152q^2 + \cdots)\\
q^{\frac{2}{3}}(54 + 14148q + 415152q^2 + \cdots)
\end{array}
\right)\\
\hline
22 & \frac{5}{3} & \frac{5}{3} & 2& 88 & \mbox{\cite{GHM}} &
q^{-11/12}\left(
\begin{array}{l}
1 + 88q + 99935q^2 + \cdots\\
q^{\frac{5}{3}}(32076 + 4185918q + 169667460q^2 + \cdots)\\
q^{\frac{5}{3}}(32076 + 4185918q + 169667460q^2 + \cdots)
\end{array}
\right)\\
\hline
\end{array}
$$


\newcommand{\etalchar}[1]{$^{#1}$}

\end{document}